\documentclass[12pt]{article}
\usepackage{amssymb,amsfonts,amsmath,amsthm,url}
\usepackage{times}
\usepackage{epsfig,graphicx,wrapfig}
\usepackage{color}
\usepackage{multirow}
\usepackage{rotating}
\usepackage{bbm}
\usepackage{latexsym}

\newtheorem{theorem}{Theorem}[section]
\newtheorem*{theorem*}{Theorem}
\newtheorem{corollary}[theorem]{Corollary}
\newtheorem{remark}[theorem]{Remark}

\newtheorem*{example*}{Example}
\newtheorem{proposition}[theorem]{Proposition}
\newtheorem{definition}[theorem]{Definition}
\newtheorem{lemma}[theorem]{Lemma}

\def\P{\mathcal{P}}

\def\od{\stackrel{\mathrm{def}}{=}}

\def\RR{\mathbb{R}}
\def\supp{\operatorname{supp}}

\def\det{\operatorname{det}}

\def\cm{\mathrm{cm}}

\def\conv#1{\operatorname{conv}\{#1\}}

\usepackage{multicol}
\usepackage{color}
\definecolor{gold}{rgb}{0.85,.66,0}
\definecolor{cherry}{rgb}{0.9,.1,.2}
\definecolor{burgundy}{rgb}{0.8,.2,.2}
\definecolor{orangered}{rgb}{0.85,.3,0}
\definecolor{orange}{rgb}{0.85,.4,0}
\definecolor{olive}{rgb}{.45,.4,0}
\definecolor{lime}{rgb}{.6,.9,0}
\definecolor{green}{rgb}{.2,.7,0}
\definecolor{darkgreen}{rgb}{.1,.5,0}
\definecolor{grey}{rgb}{.4,.4,.2}
\definecolor{brown}{rgb}{.4,.2,.1}
\definecolor{blue}{rgb}{0,.0, .81}
\definecolor{bluepurple}{rgb}{.3, .0, .7}

\def\carina#1{#1} 

\begin{document}

\begin{center}
\begin{large}
\noindent {\bf Pattern completion in symmetric threshold-linear networks}
\end{large}
\medskip

Carina Curto$^1$ and Katherine Morrison$^{1,2}$

\begin{small}
\noindent May 6, 2016
\end{small}

\end{center}

\begin{small}
\noindent $^1$ Department of Mathematics, The Pennsylvania State University, University Park, PA 16802

\noindent $^2$ School of Mathematical Sciences, University of Northern Colorado, Greeley, CO 80639 
\end{small}

\begin{abstract}
Threshold-linear networks are a common class of firing rate models that describe recurrent interactions among neurons.  Unlike their linear counterparts, these networks
generically possess multiple stable fixed points (steady states), making them viable candidates for memory encoding and retrieval.  In this work, 
we characterize stable fixed points of general threshold-linear networks with constant external drive, and discover constraints on the co-existence of fixed points involving different subsets of active neurons. In the case of symmetric networks, we prove the following {\it antichain} property:
if a set of neurons $\tau$ is the support of a stable fixed point, then no proper subset or superset of $\tau$ can support a stable fixed point.  Symmetric threshold-linear networks thus appear to be well suited for pattern completion, since the dynamics are guaranteed not to get ``stuck'' in a subset or superset of a stored pattern.  We also show that for any graph G, we can construct a network whose stable fixed points correspond precisely to the maximal cliques of G.  As an application, we design network decoders for place field codes, and demonstrate their efficacy for error correction and pattern completion.  The proofs of our main results build on the theory of permitted sets in threshold-linear networks, including recently-developed connections to classical distance geometry.
\end{abstract}

\begin{footnotesize}

\tableofcontents

\end{footnotesize}

\pagebreak

\section{Introduction}

In this work, we study stable fixed points\footnote{Equivalently: steady states, fixed point attractors, or stable equilibria.} of threshold-linear networks with constant external drive.  These networks model the activity of a population of neurons with recurrent interactions, whose dynamics are governed by the system of equations:
\begin{equation}\label{eq:network}
\dfrac{dx_i}{dt} = -x_i + \left[\sum_{j=1}^n W_{ij}x_j+\theta\right]_+, \quad i \in [n].
\end{equation}
Here $[n] = \{1,\ldots,n\}$ is the index set for $n$ neurons, $x_i(t)$ is the activity level (firing rate) of the $i^{\mathrm{th}}$ neuron, and  $\theta \in \RR$ is a constant external drive (the same for each neuron).  The real-valued $n \times n$ matrix $W$ governs the interactions between neurons, with $W_{ij}$ the effective connection strength from the $j^{\mathrm{th}}$ to the $i^{\mathrm{th}}$ neuron.  The threshold-linear function $[y]_+ = \max\{0,y\}$ ensures that $x_i(t) \geq 0$ for all $t > 0$, provided $x_i(0) \geq 0$.

The equations~\eqref{eq:network} differ from the linear system of ODEs, $\dot{x} = (-I+W)x + \theta$, only by the nonlinearity $[\;]_+$.  This nonlinearity is quite significant, however, as it allows the network to possess multiple stable fixed points, 
even though the analogous linear system can have at most one. 
Multistability is the key feature that makes threshold-linear networks viable models of memory encoding and retrieval \cite{Tsodyks, TrevesRolls, VogelsRajanAbbott, AppendixE}.  If, for example, $W$ is a symmetric matrix with nonpositive entries, then the system~\eqref{eq:network} is guaranteed to converge to a stable fixed point for any initial condition \cite{HahnSeungSlotine}.  The threshold-linear network thus functions as a traditional attractor neural network, in the spirit of the Hopfield model \cite{Hopfield1, Hopfield2}, with initial conditions playing the role of inputs and stable fixed points comprising outputs.

In this paper, we characterize stable fixed points of general threshold-linear networks, and discover constraints on the co-existence of fixed points involving different subsets of active neurons.  In the rest of this section, we give an overview of our main results, Theorems~\ref{Thm2} and~\ref{thm:symmetric-CTLN}, and explain their relevance to pattern completion.  We then illustrate their power in an application.  The remainder of the paper lays the foundation for proving the main results.
 Section~\ref{sec:prelim} summarizes relevant notation and background about permitted sets and their relationship to fixed points.  In Section~\ref{sec:gen-fixed-points} we provide general conditions that must be satisfied by fixed points of~\eqref{eq:network}.  Next, in Section~\ref{sec:geometry}, we prove a key technical result using classical distance geometry.  Finally, in Section~\ref{sec:main-proofs} we prove our main theorems by combining results from Sections~\ref{sec:gen-fixed-points} and~\ref{sec:geometry}.

\subsection{Fixed points of threshold-linear networks}
A vector $x^* \in \RR^n_{\geq 0}$ is a {\it fixed point} of~\eqref{eq:network} if, when evaluating at $x = x^*$, we obtain $dx_i/dt = 0$ for each $i$.
The {\it support} of a fixed point is given by 
$$\supp(x^*) \od \{i \in [n] \mid x^*_i > 0\}.$$ 
We use greek letters, such as $\sigma, \tau \subseteq [n]$, to denote supports.  A {\it principal submatrix} is obtained from a larger $n \times n$ matrix by restricting both row and column indices to some $\sigma \subset [n]$.  For example, $W_\sigma$ is the $|\sigma| \times |\sigma|$ principal submatrix of connection strengths among neurons in $\sigma$.  Similarly, $x_\sigma$ is the vector of firing rates for only the neurons in $\sigma$.

Although a threshold-linear network may have many stable fixed points, each fixed point is completely determined by its support.  To see why, observe that if $x^*$ is a fixed point with support $\sigma \subseteq [n]$, then
$$x_\sigma^* = [W_\sigma x^*_\sigma + \theta 1_\sigma]_+ > 0,$$
where $1_\sigma$ is a column vector of all ones.  This means we can drop the nonlinearity to obtain $(I-W_\sigma) x^*_\sigma = \theta 1_\sigma.$  If $x^*$ is a {\it stable} fixed point, then 
$I-W_\sigma$ is invertible \cite[Theorem 1.2]{flex-memory}, 
and hence we can solve for $x^*_\sigma$ explicitly as
$$x^*_\sigma = \theta(I-W_\sigma)^{-1} 1_\sigma.$$
Of course, this is only a fixed point of~\eqref{eq:network} if $x_\sigma^*>0$ and 
\carina{$x_k^* = [\sum_{i\in\sigma} W_{ki}x_i^*+\theta]_+ = 0$} for all $k \notin \sigma$.  If either of these conditions fail, then the fixed point with support $\sigma$ does not exist.  If it does exist, however, the above formula gives the precise values of the nonzero entries $x_\sigma^*$, and guarantees that it is unique.
Thus, in order to understand the stable fixed points of a network~\eqref{eq:network} it suffices to characterize the possible supports.  

\subsection{Summary of main results}
Given a choice of $W$ and $\theta$, what are the possible stable fixed point supports?  In this paper,
we provide a set of conditions that fully characterize these supports for general threshold-linear networks, and show how the conditions simplify when $W$ is inhibitory or symmetric  (see Section~\ref{sec:gen-fixed-points}).
The compatibility of the fixed point conditions across multiple $\sigma \subseteq [n]$ enables us to obtain results about which collections $\{\sigma_i\}$ of supports can (or cannot) co-exist in the same network.  Our strongest results in this vein arise when specializing to symmetric $W$.  In this case, we can build on the geometric theory of permitted sets, as introduced in \cite{net-encoding}, in order to greatly constrain the collection of allowed fixed point supports.  The following is our first main result.

\begin{theorem}[Antichain property]\label{Thm2}
Consider the threshold-linear network~\eqref{eq:network}, for a symmetric matrix~$W$ with zero diagonal. If there exists a stable fixed point with support $\tau \subseteq [n]$, then there is no stable fixed point with support $\sigma$ for any  $\sigma \subsetneq \tau$ or $\sigma \supsetneq \tau$.
\end{theorem}

\noindent The proof is given in Section~\ref{sec:main-proofs}, and relies critically on a technical result, Proposition~\ref{prop:antichain}, which we state and prove in Section~\ref{sec:geometry} using ideas from classical distance geometry.

Theorem~\ref{Thm2} has several immediate consequences.  First, it implies that the set of possible fixed point supports is an antichain in the Boolean lattice $(2^{[n]}, \subseteq),$ where an {\it antichain} is defined as a set of incomparable elements in the poset (that is, no two elements are related by $\subseteq$).  Sperner's theorem \cite{sperner} states that the cardinality of a maximum antichain in the Boolean lattice is precisely ${n \choose \lfloor n/2 \rfloor}$ (see  Figure~\ref{fig:antichain}). We thus have the corollary:

\begin{corollary}\label{cor:sperner}
If $W$ is symmetric with zero diagonal, then the threshold-linear network~\eqref{eq:network} has at most ${n \choose \lfloor n/2 \rfloor}$ stable fixed points.
\end{corollary}

We note, however, that this upper bound is not necessarily tight.  We currently do not know whether or not there exist networks with ${n \choose \lfloor n/2 \rfloor}$ stable fixed points for each $n$. More generally, it is an open question to determine which antichains in the Boolean lattice can be realized as the set of fixed point supports for a symmetric $W$.

\begin{figure}[h!]
\begin{center}
\includegraphics[width=2in]{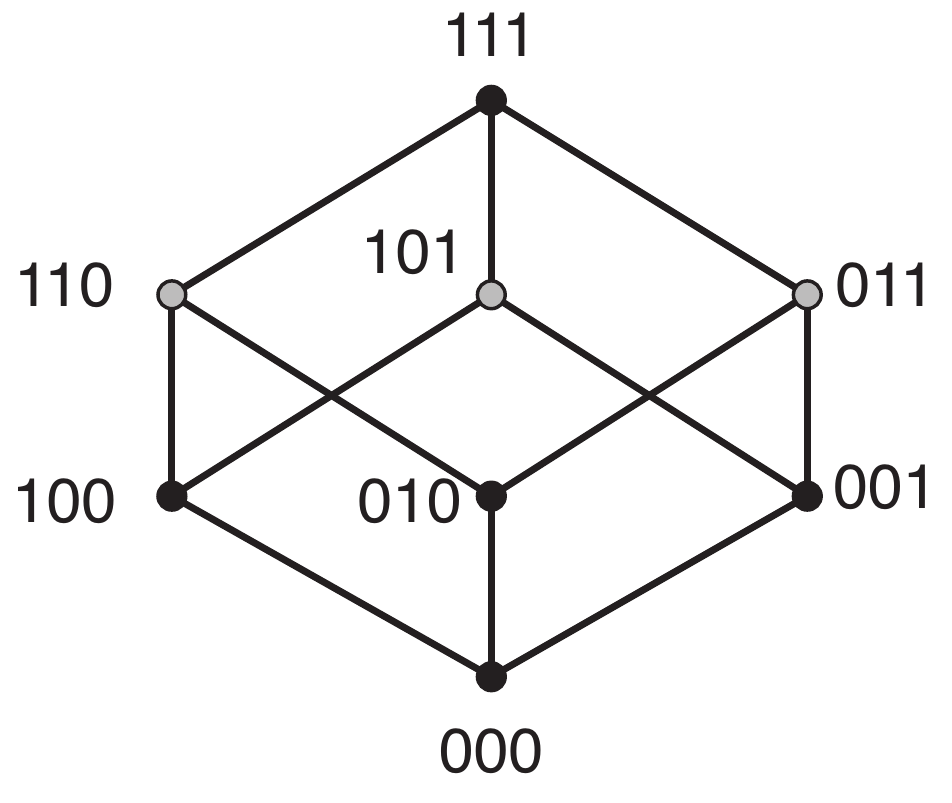}
\end{center}
\vspace{-.2in}
\caption{A Boolean lattice for $n=3$.  A maximal antichain, $\sigma_1 = \{1,2\}, \sigma_2 = \{1,3\},$ and $\sigma_3 = \{2,3\}$, is shown in gray.  This antichain has size ${3 \choose \lfloor 3/2 \rfloor} = 3$, in agreement with Sperner's theorem.}
\label{fig:antichain}
\end{figure}

A second consequence of Theorem~\ref{Thm2} is that symmetric threshold-linear networks appear well suited for pattern completion.  A network can perform pattern completion if, after initializing at a subset of a pattern, the network dynamics evolve the activity to the complete pattern.
In this context, a {\it pattern} of the network is a subset of neurons corresponding to the support of a stable fixed point.  Theorem~\ref{Thm2} implies that if $\tau \subseteq [n]$ is a pattern of a symmetric network, then the network activity is guaranteed not to ``get stuck'' in a subpattern $\sigma \subsetneq \tau$ or a superpattern $\sigma \supsetneq \tau$.

In order to further illustrate the power of Theorem~\ref{Thm2}, we now turn to the special case of symmetric 
threshold-linear networks with binary synapses.  Here the connection matrix $W=W(G, \varepsilon, \delta)$ is specified by a simple graph\footnote{Recall that a {\it simple} graph has undirected edges, no self-loops, and at most one edge between each pair of vertices.} $G$, whose vertices correspond to neurons:
{\begin{equation}\label{eq:binary-synapse}
W_{ij} = \left\{\begin{array}{cc} 0 & \text{ if } i = j, \\ -1 + \varepsilon & \text{ if } (ij) \in G,\\ -1 -\delta & \text{ if } (ij) \notin G. \end{array}\right.
\end{equation}}
\hspace{-0.05in}The parameters $\varepsilon, \delta \in \RR$ satisfy $0<\varepsilon<1$ and $\delta >0$, while $(ij) \in G$ indicates that there is an edge between vertices $i$ and $j$.  Note that since $\varepsilon < 1$, the network is inhibitory ($W_{ij}\leq 0$).  \carina{One might} interpret $W_{ij}$ as the effective connection strength between neurons $i$ and $j$ due to competitive inhibition that is attenuated by excitation whenever $(ij) \in G$ (see Figure~\ref{fig:network-cartoon}).  

\begin{figure}[h!]
\centering{
\includegraphics[width=4.5in]{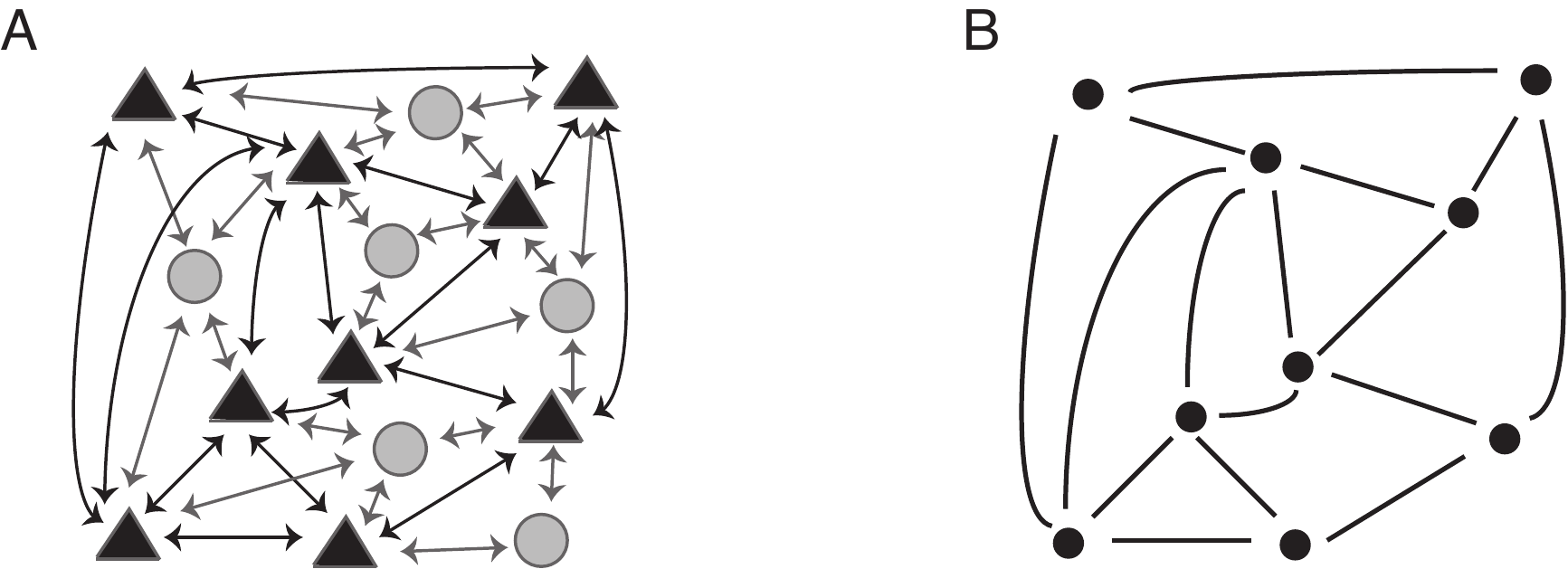}
}
\footnotesize{\caption{
Excitatory connections in a sea of inhibition.  ({\bf A}) Pyramidal neurons (black triangles) have bidirectional synapses onto each other.  Inhibitory interneurons (gray circles) produce non-specific inhibition to all neighboring pyramidal cells.  ({\bf B}) The graph of the network in (A), retaining only excitatory neurons and their connections.  A missing edge corresponds to strong competitive inhibition, $-1-\delta$, while the presence of an edge indicates attenuated inhibition, $-1+\varepsilon$, resulting from the sum of global background inhibition and an excitatory connection.
}
\label{fig:network-cartoon}
}
\end{figure}

In the case of  binary symmetric networks, combining our characterization of fixed point supports with Theorem~\ref{Thm2} we obtain our second main result.
To state it, we need some standard graph-theoretic terminology.  A subset of vertices $\sigma$ is a {\em clique} of $G$ if $(ij) \in G$ for all pairs $i,j \in \sigma$.  In other words, a clique is a subset of neurons that is all-to-all connected.  A clique $\sigma$ is called {\it maximal} if it is not contained in any larger clique of $G$.

\begin{theorem} \label{thm:symmetric-CTLN}
Let $G$ be a simple graph, and consider the network~\eqref{eq:network} with $W = W(G,\varepsilon, \delta)$ for any $0 < \varepsilon <1$, $\delta>0$ and $\theta>0$.  The stable fixed points of this network are in one-to-one correspondence with the maximal cliques of $G$.  Specifically, each maximal clique $\sigma$ is the support of a stable fixed point, given by
$$x_\sigma^* = \dfrac{\theta}{(1-\varepsilon)|\sigma|+\varepsilon}1_\sigma,$$
and there are no other stable fixed points.
\end{theorem}

\noindent The proof of Theorem~\ref{thm:symmetric-CTLN} is given in Section~\ref{sec:main-proofs}, and demonstrates that it is possible to have a network in which all stable fixed point supports correspond to maximal stored patterns.  This situation is ideal for pattern completion.  We then show that this feature generalizes to a much broader class of symmetric networks (see Theorem~\ref{Thm3} in Section~\ref{sec:main-proofs}).

\carina{Theorem~\ref{thm:symmetric-CTLN} provides additional insight into the question of how many stable fixed points can be stored in a symmetric threshold-linear network of $n$ neurons.  Corollary~\ref{cor:sperner} gave us an upper bound of ${n \choose \lfloor n/2 \rfloor},$ as a consequence of Sperner's theorem and Theorem~\ref{Thm2}.  As a corollary of Theorem~\ref{thm:symmetric-CTLN}, we now obtain a lower bound.}

\begin{corollary} \carina{For each $n$, there exists a symmetric threshold-linear network of the form~\eqref{eq:network} with $2^{\lfloor n/2 \rfloor}$ stable fixed points.}
\end{corollary}

\carina{To see how this follows from  Theorem~\ref{thm:symmetric-CTLN}, let $m = \lfloor n/2 \rfloor$ and consider the complete $m$-partite graph $G$, where each part has $2$ vertices.  (Note that two vertices of $G$ have an edge between them if and only if they belong to distinct parts.)  If $n$ is odd, let the last vertex connect to all the others.  It is easy to see that $G$ has precisely $2^m$ maximal cliques, obtained by selecting one vertex from each part.  By Theorem~\ref{thm:symmetric-CTLN}, the corresponding network with $W = W(G,\varepsilon,\delta)$ has $2^m$ stable fixed points.}

We end this section with a specific application of Theorem~\ref{thm:symmetric-CTLN}, illustrating the capability of binary symmetric networks to correct noise-induced errors in combinatorial neural codes. 

\subsection{Application: pattern completion and error correction for place field codes}\label{sec:application}

The hippocampus contains a special class of neurons, called place cells, that act as position sensors for tracking the animal's location in space \cite{OKeefe}.  Place fields are the spatial receptive fields associated to place cells, and {\it place field codes} are the corresponding combinatorial neural codes defined from intersections of place fields \cite{neuro-coding}.  In this application of Theorem~\ref{thm:symmetric-CTLN} we examine the performance of a binary symmetric network that is designed to perform pattern completion and error correction for place field codes.  
\carina{Our aim is to illustrate how a simple threshold-linear network could function as an effective decoder for a biologically realistic neural code.  We do not wish to suggest, however, that the specific networks considered in Theorem~\ref{thm:symmetric-CTLN} are in any way accurate models of the hippocampus.}

\begin{figure}[!h]
\begin{center}
\includegraphics[width=5.75in]{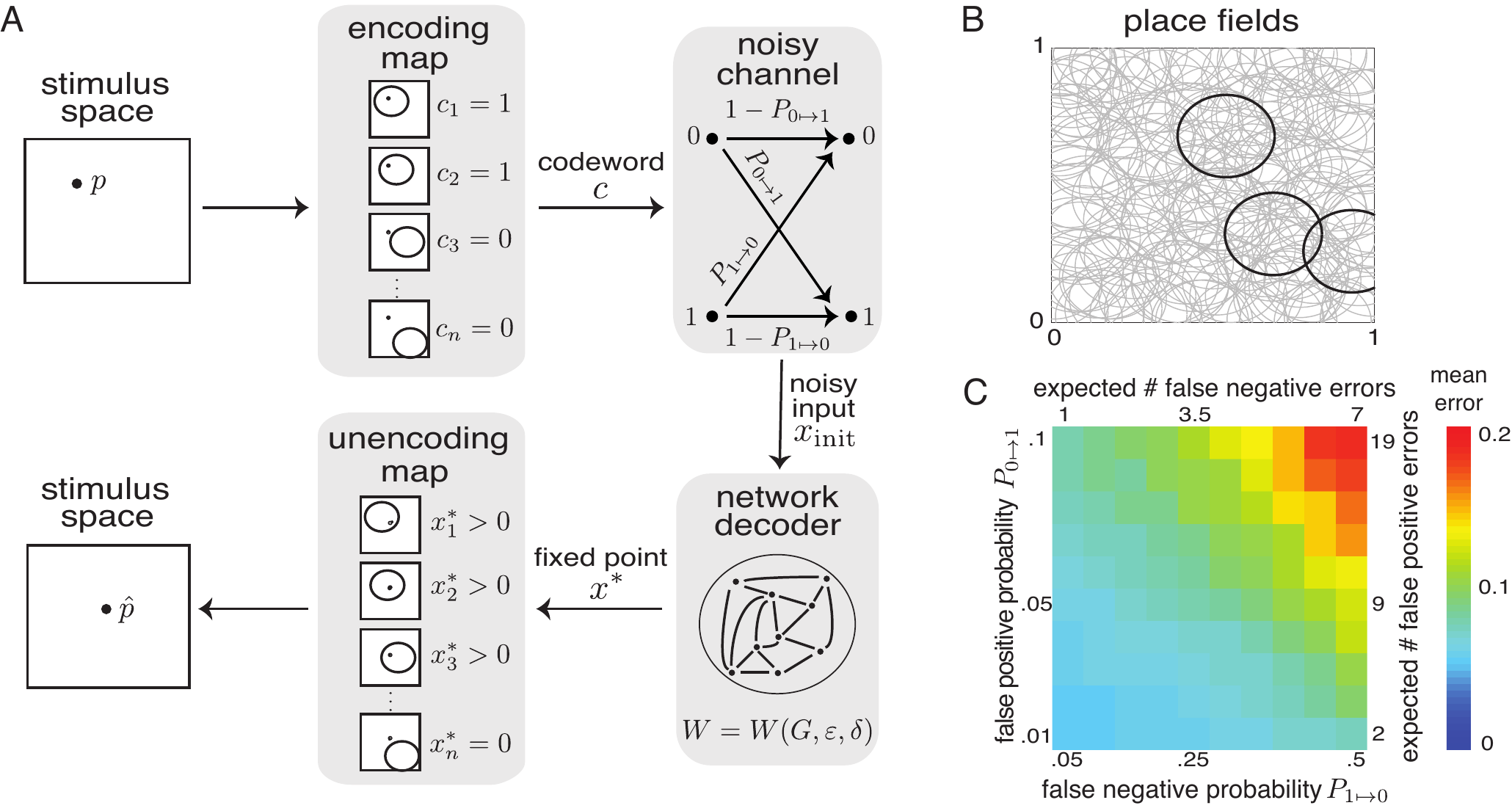}
\end{center}
\vspace{-.3in}
\caption{The binary symmetric network decoder has low mean distance error.  \textbf{(A)} The sequence of steps in one trial of the binary symmetric network decoder.  A point $p$ is chosen uniformly at random from a two-dimensional stimulus space, and is passed through an encoding map defined by the set of place fields, resulting in a codeword $c$.  The codeword then goes through a noisy channel, and the corrupted word $x_{\mathrm{init}}$ becomes the input (initial condition) for the binary symmetric network decoder.  The network evolves to a fixed point $x^*$, which is then passed to the unencoding map.  The final output is an estimated point in the stimulus space, $\hat{p}$.  \textbf{(B)} 200 place fields covering a 1x1 stimulus space, yielding a place field code of length 200 with an average of 14 neurons (7\%) firing per codeword.   Three place fields are highlighted in black for clarity.  \textbf{(C)} Performance of the code from (B) across a range of noisy channel conditions 
$(P_{1\mapsto 0}, P_{0\mapsto 1})$.  For each condition, 1,000 trials were performed following the process described in (A).  The distance error was calculated as $\Vert p - \hat{p} \Vert$.  Colors denote mean distance errors across 1,000 trials.}
\label{fig:PF-decoder}
\end{figure}

We generate place field codes from a set of circular place fields, $U_1,\ldots, U_n \subset [0,1]^2$, in a square box environment (see Figure~\ref{fig:PF-decoder}B).  Each position $p \in [0,1]^2$ corresponds to a binary codeword, $c = c_1 c_2 \cdots c_n$, where
$$c_i=\left\{\begin{array}{ll} 1& \textrm{if } p \in U_i\\ 0& \textrm{if } p \notin U_i \end{array}\right..$$
The binary symmetric network corresponding to such a code has $n$ neurons, one for each place \carina{field}, and assigns the larger weight, $-1+\varepsilon$, to \carina{connections} between neurons whose place fields overlap:
\begin{eqnarray}\label{eq:W-placefield}
W_{ij} &=& \left\{\begin{array}{cc} 0 & \text{ if } i = j, \\ -1 + \varepsilon & \text{ if } U_i \cap U_j \neq \emptyset,\\ -1 -\delta & \text{ if } U_i \cap U_j = \emptyset. \end{array}\right.
\end{eqnarray}
This is precisely the matrix $W(G,\varepsilon,\delta)$, defined in~\eqref{eq:binary-synapse}, where $G$ is the co-firing graph with edges 
     $(ij) \in G$ if there exists a codeword $c$ such that $c_i=c_j=1$.  It is easy to see that the network $W$ can be learned from a simple rule where each \carina{connection} $W_{ij}$ is initially set to $-1-\delta$, and then potentiated to $-1+\varepsilon$ after presentation of a codeword in which the pair of neurons $i$ and $j$ co-fire.

The above network can be used to correct errors induced by transmitting codewords of the place field code through a noisy channel. 
Figure~\ref{fig:PF-decoder}A  shows the basic paradigm of error correction by a network decoder.  A point $p$ in the stimulus space (the animal's square box environment) is encoded as a binary codeword via an encoding map given by the place fields $\{U_1, \ldots, U_n\}$.  This codeword is then passed through a noisy channel, and is received by the network decoder as a noisy initial condition.  The network then evolves according to~\eqref{eq:network}, until it reaches  a stable fixed point $x^*$ (see Appendix~\ref{sec:gen-PF-codes} for further details).
   From $x^*$ we obtain an estimated point $\hat{p}$ in the stimulus space, given by the mean of the centers of all place fields $U_i$ such that $x_i^*>0$.   The \emph{distance error} of the network decoder on a single trial is the Euclidean distance $\Vert p-\hat{p} \Vert$.

To test the performance of the binary symmetric network decoder, we randomly generated place field codes with 200 neurons, and an average of $\approx 7$\% of neurons firing per codeword (see Appendix~\ref{sec:gen-PF-codes} for further details). Figure~\ref{fig:PF-decoder}B illustrates the coverage of the $1\times 1$ square box environment by 200 place fields.  For each code, we computed $W$ according to~\eqref{eq:W-placefield}, with  $\varepsilon = 0.25$ and  $\delta = 0.5,$ to obtain the corresponding network decoder.  We then performed 1,000 trials for each of 100 different noisy channel conditions.  Each noise condition consisted of a \emph{false positive probability} $P_{0\mapsto 1}$, the probability that a 0 is flipped to a 1 in a transmitted codeword, and a \emph{false negative probability} $P_{1 \mapsto 0}$, the probability of a $1 \mapsto 0$ flip (see Figure~\ref{fig:PF-decoder}A, top right).  
Figure~\ref{fig:PF-decoder}C shows the mean distance errors of the network decoder, for a range of noise conditions $(P_{1 \mapsto 0}, P_{0\mapsto 1})$.
Note that because the expected number of 1s in each codeword is $\approx 14$ (out of 200 bits), a value of $P_{0\mapsto 1} = 0.1$ yields nearly 19 expected false positive errors per codeword, while $P_{1 \mapsto 0} = 0.5$ results in only 7 expected false negative errors per codeword.  Most of the considered noise conditions yielded mean distance errors of $0.1$ or less.  Even for very severe noise conditions with large numbers of expected errors, the mean distance error did not exceed $0.2$, or about 20\% of the side length of the environment.
These results are representative of all place field codes tested.


\section{Preliminaries}\label{sec:prelim}

Here we introduce some notation and review some basic facts about permitted sets and fixed points, including the connection to distance geometry in the symmetric case.

\subsection{Notation}

Using vector notation, $x = (x_1,\ldots,x_n)^T$, we can rewrite equation~\eqref{eq:network} more compactly as
$$\dot{x} = -x + [Wx + \theta]_+,$$
where $\dot{x}$ denotes the time derivative $dx/dt$, and the nonlinearity $[ \; ]_+$ is applied entry-wise.  Solutions to~\eqref{eq:network} have the property that $x(t) \geq 0$ for all $t>0$, provided that $x(0) \geq 0$.
We use the notation $x \geq 0$ to indicate that every entry of the vector $x$ is nonnegative.  The inequalities $>, <,$ and $\leq$ are similarly applied entry-wise to vectors. 
If $\sigma = \supp(x)$, then the restriction of $x$ to its support satisfies $x_\sigma > 0$, while $x_k = 0$ for all $k \notin \sigma$.  

Many of our results are most conveniently stated in terms of the auxiliary matrix 
\begin{equation}\label{eq:A}
A = 11^T - I + W,
\end{equation}
where $11^T$ is the rank one matrix with entries all equal to 1.  In other words,
\begin{equation*}
W_{ij} = \left\{\begin{array}{cc} A_{ij} & \text{for} \; i = j,\\ -1+A_{ij} & \text{for} \; i \neq j.\end{array}\right.
\end{equation*}
Recall that for any $\sigma \subseteq [n]$, the {\it principal submatrix} of $A$ obtained by restricting both rows and columns to the index set $\sigma$ is denoted $A_\sigma$.  
When $A_\sigma$ is invertible, we define the number
\begin{equation}\label{eq:a_sigma}
a_\sigma \od \sum_{i,j \in \sigma} (A_\sigma^{-1})_{ij} = - \dfrac{\cm(A_\sigma)}{\det(A_\sigma)},
\end{equation}
where the notation $A_\sigma^{-1} = (A_\sigma)^{-1} \neq (A^{-1})_\sigma$.  
The second equality is a simple consequence of Cramer's Rule, and connects $a_\sigma$ to the Cayley-Menger determinant,
$$\cm(A) \od \left(\begin{array}{cc} 0 & 1^T \\ 1 & A\end{array}\right),$$
defined for any $n \times n$ matrix $A$.   Here $1$ denotes the $n \times 1$ column vector of all ones, and $1^T$ is the corresponding row vector.  Note that $a_\sigma$ has nice scaling properties: if $A_\sigma \mapsto tA_\sigma$, then $a_\sigma \mapsto \dfrac{1}{t} a_\sigma.$  

The Cayley-Menger determinant is well-known for its geometric meaning, which we will review below.  It also appears in the determinant formula (see \cite[Lemma 7]{net-encoding}),
\begin{equation}\label{eq:det-formula}
\det(-11^T+A) = \det(A) + \cm(A),
\end{equation}
which holds for any $n \times n$ matrix $A$.  We will make use of this formula in the proof of Theorem~\ref{thm:fixed-pts}.

\subsection{Permitted sets and fixed points}\label{sec:prelim-permitted}
Instead of restricting to a single external drive $\theta$ for all neurons,  prior work has considered a generalization of the network~\eqref{eq:network} of the form
\begin{equation}\label{eq:permitted}
\dot{x} = -x + [Wx + b]_+,
\end{equation}
where $b \in \RR^n$ is an arbitrary vector of external drives $b_i$ for each neuron.
In this context, Hahnloser, Seung, and collaborators \cite{HahnSeungSlotine, XieHahnSeung} developed the idea of \emph{permitted sets} of the network, which are the supports of stable fixed points that can arise for some $b$.  The set of all permitted sets of a network of the form~\eqref{eq:permitted} thus depends only on $W$, and is denoted $\P(W)$.   We know from prior work \cite{HahnSeungSlotine, flex-memory}
that $\sigma$ is a permitted set if and only if $(-I+W)_\sigma$ is a stable matrix.\footnote{A matrix is {\it stable} if all of its eigenvalues have strictly negative real parts.}  In other words,
\begin{eqnarray}\label{eq:P(W)}
\P(W) &=& \{\sigma \subseteq [n] \mid -11^T+A_\sigma \text{ is a stable matrix}\},
\end{eqnarray}
since $(-I+W)_\sigma = -11^T+A_\sigma$.

In the special case considered here, where $b_i = \theta$ for all $i \in [n]$, 
it is clear that any stable fixed point $x^*$ of~\eqref{eq:network} must have $\supp(x^*) \in \P(W)$.
The converse, however, is not true.
Although all permitted sets have corresponding fixed points for some $b \in \RR^n$, there is no guarantee that such a fixed point exists for $b = \theta$.  For example, in the case of symmetric $W$ it has been shown that $\sigma \in \P(W)$ implies $\tau \in \P(W)$ for all subsets $\tau \subset \sigma.$  
This property appears to imply that such networks cannot perform pattern completion, since subsets of permitted sets are also permitted.  However, Theorem~\ref{Thm2} tells us that if a permitted set is the support of a stable fixed point of~\eqref{eq:network}, none of its subsets can be a fixed point support in the case $b = \theta$, despite being permitted.

\subsection{Permitted sets for symmetric $W$} \label{sec:geometry-background}
In order to prove Theorem~\ref{Thm2}, we will make heavy use of the geometric theory of permitted sets that was developed in \cite{net-encoding}.  Here we review some basic facts from classical distance geometry that allow us to geometrically characterize the permitted sets of a network when $W$ is symmetric.

An $n\times n$ matrix $A$ is a \emph{nondegenerate square distance matrix} if there exists a configuration of points $\{p_i\}_{i \in [n]}$ in the Euclidean space $\RR^{n-1}$ such that $A_{ij} = ||p_i-p_j||^2$, and the convex hull of the $p_i$s forms a full-dimensional simplex.  The Cayley-Menger determinant $\cm(A)$ computes the volume of this simplex, and
can be used to detect whether or not a given matrix is nondegenerate square distance.   In particular, 
if $A_\sigma$ is a nondegenerate square distance matrix with $|\sigma|>1$, then $A_\sigma$ is invertible and $a_\sigma>0$ \cite[Corollary 8]{net-encoding}. See \cite[Appendix A]{net-encoding} for a more complete review of these and other related facts about nondegenerate square distance matrices.

In the singleton case, $\sigma = \{i\}$ for some $i \in [n]$, the matrix $A_\sigma = [0]$ is always a nondegenerate square distance matrix with $\cm(A_\sigma) = -1 \neq 0$, although $\det(A_\sigma) = 0$.  Because of this, it is convenient to declare $a_\sigma = a_{\{i\}} = \infty$ whenever $A_{ii} = 0$.  With this convention, we can state the following geometric characterization of permitted sets for symmetric networks, first given in \cite{net-encoding}. 

\begin{lemma}[Proposition 1, \cite{net-encoding}] \label{lemma:geometric-char}
Let $W$ be a symmetric network with zero diagonal, let $A = 11^T - I + W$, as in~\eqref{eq:A}, and suppose $\sigma \subseteq [n]$ is nonempty.  Then $\sigma \in \P(W)$ if and only if $A_\sigma$ is a nondegenerate square distance matrix with $a_\sigma>1.$  Equivalently, $-11^T + A_\sigma$ is a stable matrix if and only if $A_\sigma$ is a nondegenerate square distance matrix with $a_\sigma>1.$
\end{lemma}

\section{Fixed point conditions}\label{sec:gen-fixed-points}

In this section, we derive general fixed point conditions for networks of the form~\eqref{eq:network}, and demonstrate simplifications of these conditions for inhibitory networks and symmetric networks.  

\subsection{General fixed point conditions}

Consider a fixed point $x^*$ with $\supp(x^*) = \sigma$.   
The conditions for $x^*$ to be a fixed point of~\eqref{eq:network} are given by the equation:
$$x^* = [Wx^*+\theta]_+,$$
together with the requirement that $x_i^*>0$ for all $i \in \sigma$, and $x_k^*=0$ for all $k \notin \sigma$.
 For a single neuron $i \in [n]$, the fixed point equation becomes:
\begin{eqnarray*}
x_i^* &=& \left[\sum_{j\in\sigma} W_{ij} x_j^* + \theta \right]_+ = 
\left[x_i^* - \sum_{j\in\sigma} (11^T)_{ij}x_j^* + \sum_{j\in\sigma} A_{ij}x_j^* + \theta \right]_+\\
&=&  \left[x_i^* - m(x^*) + (A x^*)_i + \theta \right]_+,
\end{eqnarray*}
where $A$ is given by~\eqref{eq:A} and 
$$m(x^*) \od \sum_{i=1}^n x_i^* = \sum_{i\in\sigma} x_i^*$$ 
is the total population activity.
Separating out the ``on'' neurons $i \in \sigma$ from the ``off'' neurons $k \notin \sigma$, we obtain the following lemma.

\begin{lemma}\label{lemma:fixed-pts}
Consider a vector $x^*$, with $ \supp(x^*) = \sigma.$ 
Then $x^*$ is a fixed point of~\eqref{eq:network} if and only if
\begin{eqnarray*}
(Ax^*)_i  &=& m(x^*) - \theta, \quad \text{for all} \;\; i \in \sigma, \;\; \text{and}\\
(Ax^*)_k &\leq& m(x^*) - \theta, \quad \text{for all} \;\; k \notin \sigma. 
\end{eqnarray*}
In particular,
$$(Ax^*)_k \leq (Ax^*)_i  = (Ax^*)_j  \quad \text{for all} \;\; i,j \in \sigma, \; k \notin \sigma.$$
\end{lemma}

As an immediate consequence, we can list conditions that must be satisfied for fixed points consisting of a single active neuron $i$, with firing rate $x_i^*>0$.  In this case,
$m(x^*) = x_i^*$, and so by Lemma~\ref{lemma:fixed-pts},
$$(Ax^*)_i = A_{ii}x_i^* = x_i^* - \theta, \quad \text{and} \quad (Ax^*)_k = A_{ki}x_i^* \leq A_{ii} x_i^*, \text{ for all } k \neq i.$$
This allows us to solve for $x_i^* = \theta/(1-A_{ii}),$ and to conclude $A_{ki} \leq A_{ii}$.  In order for this fixed point
to be {\it stable}, we must also have $A_{ii} < 1$, and hence $\theta > 0$ because $x_i^* > 0$.  We collect these observations
in the following proposition, together with the case of empty support, $x^*=0$.

\begin{proposition}\label{prop:singleton}
Suppose $|\sigma| \leq 1$.  If $\sigma = \{i\}$, then there exists a stable fixed point $x^*$ of ~\eqref{eq:network} with support $\sigma$ if and only if the following all hold:
\begin{itemize}
\item[(i)] $A_{ii} < 1$ (equivalently, $W_{ii} < 1$),
\item[(ii)] $\theta > 0$, and
\item[(iii)] $A_{ki} \leq A_{ii}$ (equivalently, $W_{ki} \leq -1+W_{ii}$) for all $k \neq i$.
\end{itemize}
Moreover, if this stable fixed point exists, then it is given by $x_k^* = 0$ for all $k \neq i$ and
$$x_i^* =  \dfrac{\theta}{1-A_{ii}}.$$
Alternatively, if $\sigma = \emptyset$, then $x^* = 0$ is a fixed point of~\eqref{eq:network} if and only if $\theta \leq 0$.  If $\theta < 0$, then this fixed point is guaranteed to be stable.
\end{proposition}

\noindent Analogous conditions can be obtained for $|\sigma|>1,$ so long as $A_{\sigma}$ is not ``fine-tuned,'' as defined below.

\begin{definition} \normalfont
We say that a principal submatrix $A_\sigma$ of $A$ is {\em fine-tuned} if either of the following are true:
\begin{itemize}
\item[(a)] $\det(A_\sigma) = 0$, or
\item[(b)] $\sum_{i,j\in\sigma} A_{ki} (A_\sigma^{-1})_{ij} = 1$ for some $k \notin \sigma.$ 
\end{itemize}
Note that (b) depends on entries of the full matrix $A,$ not just entries of $A_\sigma$.
\end{definition}

\noindent We can now state and prove Theorem~\ref{thm:fixed-pts}.  

\begin{theorem}[General fixed point conditions]\label{thm:fixed-pts}
Consider the threshold-linear network~\eqref{eq:network}, 
and let $A = 11^T - I + W$, as in~\eqref{eq:A}.  Suppose 
$\sigma \subseteq [n]$ is nonempty and $A_\sigma$ is not fine-tuned.
Then there exists a stable fixed point with support $\sigma$ if and only if $\theta \neq 0$ and the following three conditions hold:
\begin{itemize}
\item[(i)]  (permitted set condition) $-11^T + A_\sigma$ is a stable matrix.
\item[(ii)] (on-neuron conditions) There are two cases, depending on the sign of $\theta$:
\begin{itemize}
\item[(a)] $\theta > 0$: either $a_\sigma < 0$ and $A_\sigma^{-1} 1_\sigma < 0$, or $a_\sigma> 1$ and $A_\sigma^{-1} 1_\sigma > 0.$
\item[(b)] $\theta < 0$: $0 < a_\sigma < 1$ and $A_\sigma^{-1} 1_\sigma > 0.$
\end{itemize}
\item[(iii)]  (off-neuron conditions)  For each $k \notin \sigma$, 
$\displaystyle \dfrac{\theta}{a_\sigma -1}\sum_{i,j\in\sigma} A_{ki}(A_\sigma^{-1})_{ij} < \dfrac{\theta}{a_\sigma -1}.$
\end{itemize}
Moreover, if a stable fixed point $x^*$ with $\supp(x^*)=\sigma$ exists, then it is given by
$$x_\sigma^* = \dfrac{\theta}{a_\sigma-1} A_\sigma^{-1} 1_\sigma,$$
with total population activity
$$m(x^*) = \dfrac{a_\sigma}{a_\sigma -1}\theta.$$
\end{theorem}

\begin{remark}\normalfont
For fixed points supported on a single neuron, $\sigma = \{i\}$, we may have $A_{ii} = 0$ so that $A_\sigma$ is not invertible and is thus not covered by the theorem.
This case, together with the case $\sigma = \emptyset$, 
is covered by Proposition~\ref{prop:singleton}.
\end{remark}

\begin{proof}[Proof of Theorem~\ref{thm:fixed-pts}]
($\Rightarrow$) Suppose there exists a stable fixed point $x^*$ with nonempty support, $ \supp(x^*) = \sigma.$  Then $\sigma$ is a permitted set and $(-11^T+A)_\sigma$ is a stable matrix (see equation~\eqref{eq:P(W)}), so condition (i) holds.  Next, observe that by Lemma~\ref{lemma:fixed-pts} we have
$(A x^*)_\sigma = A_\sigma x_\sigma^* = (m(x^*)-\theta)1_\sigma$. Since $\det(A_\sigma)\neq 0$ because $A_\sigma$ is not fine-tuned, we can write 
$$x_\sigma^* = (m(x^*)-\theta)A_\sigma^{-1} 1_\sigma.$$  
Summing the entries of $x_\sigma^*$ we obtain 
\begin{equation}\label{eq:m}
m(x^*) = (m(x^*)-\theta) a_\sigma.
\end{equation}
Since $m(x^*)>0,$ we must have $a_\sigma \neq 0$. 
To see that $\theta \neq 0$, note that if $\theta = 0$, then $a_\sigma = 1$ and hence $\det(A_\sigma)+\cm(A_\sigma) = 0.$  Using the determinant formula~\eqref{eq:det-formula},
we see that this implies $\det(-11^T+A_\sigma) = 0$, which contradicts the fact that 
$-11^T+A_\sigma$ is a stable matrix.
We can thus conclude that $\theta \neq 0$, which in turn implies that $a_\sigma \neq 1$,
using equation~\eqref{eq:m}.
This allows us to solve for $m(x^*)$ and $x^*$ as:
$$m(x^*) = \dfrac{a_\sigma}{a_\sigma -1}\theta, \quad \text{and} 
\quad x_\sigma^* = \dfrac{\theta}{a_\sigma -1} A_\sigma^{-1} 1_\sigma,$$
yielding the desired equations for $x_\sigma^*$ and $m(x^*)$.

To show that condition (ii) holds, we split into two cases depending on the sign of $\theta$.
Case 1: $\theta > 0$.  Because $m(x^*)>0$, we must have $\dfrac{a_\sigma}{a_\sigma-1}>0$, which implies
either $a_\sigma < 0$ or $a_\sigma>1$.  Since $x_\sigma^*>0$, in the $a_\sigma<0$ case we must have
$A_\sigma^{-1} 1_\sigma<0$, and in the $a_\sigma > 1$ case $A_\sigma^{-1} 1_\sigma>0.$
Case 2: $\theta < 0$.  Now $m(x^*)>0$ implies $\dfrac{a_\sigma}{a_\sigma-1}<0$, which yields
$0 < a_\sigma < 1$.  Since $x_\sigma^*>0$, we must have $A_\sigma^{-1} 1_\sigma>0.$
Altogether, these observations give condition (ii).

Finally, we find that for $k \notin \sigma$,
$$(Ax^*)_k = \sum_{i \in \sigma} A_{ki} x_i^* = \dfrac{\theta}{a_\sigma -1} \sum_{i\in\sigma} A_{ki} \sum_{j\in\sigma} (A_\sigma^{-1})_{ij} \leq m(x^*)-\theta=\dfrac{\theta}{a_\sigma -1} ,$$
where the inequality follows because $(Ax^*)_k \leq m(x^*)-\theta$ by Lemma~\ref{lemma:fixed-pts}.
Because $A_\sigma$ is not fine-tuned, this inequality must be strict, yielding condition (iii).

($\Leftarrow$) Now, suppose $\theta \neq 0$, and conditions (i)-(iii) all hold for a given $\sigma$, with $A_\sigma$ not fine-tuned.  Consider the ansatz:
$$x_\sigma^* = \dfrac{\theta}{a_\sigma -1} A_\sigma^{-1} 1_\sigma,$$
with $x_k^* = 0$ for $k \notin \sigma$.
By (i), $\sigma$ is a permitted set and hence any fixed point with support $\sigma$ is guaranteed to be stable, because $A_\sigma$ does not satisfy condition (b) of the definition of fine-tuned (this follows from Proposition 4 of \cite{net-encoding}).  
It thus suffices to check that $x_\sigma^*>0$ and that $x^*$ satisfies the fixed point conditions in Lemma~\ref{lemma:fixed-pts}.  Since by condition (ii) we assume in all cases that  $\dfrac{\theta}{a_\sigma -1} A_\sigma^{-1} 1_\sigma>0,$ clearly $x_\sigma^*>0$.  Moreover, it is easy to see that the fixed point conditions in Lemma~\ref{lemma:fixed-pts} are satisfied, using our assumption that condition (iii) holds and the fact that $m(x^*)-\theta = \dfrac{\theta}{a_\sigma -1}$.
\end{proof}

\subsection{Inhibitory fixed point conditions}

In special cases, the fixed point conditions can be stated in simpler terms than what we had in Theorem~\ref{thm:fixed-pts}.  Here we consider the case where $W$ is {\it inhibitory}, so that $W_{ij} \leq 0$ for all $i,j \in [n]$.  In Section~\ref{sec:symm-fixed-point}, we will see a very similar simplification when $W$ is symmetric.

\begin{theorem}\label{Thm0}
Consider the threshold-linear network~\eqref{eq:network}, with $W$ satisfying $-1 \leq W_{ij} \leq 0$ for all $i,j \in [n]$.  Let $A = 11^T - I + W$,
as in~\eqref{eq:A}.
If $\theta \leq 0$, then $x^*=0$ is the unique stable fixed point of~\eqref{eq:network}.  
If $\sigma \subseteq [n]$ is nonempty and $A_\sigma$ is not fine-tuned, then
there exists a stable fixed point with support $\sigma$ if and only if $\theta > 0$ and the following three conditions hold:
\begin{itemize}
\item[(i)]  (permitted set condition) $-11^T + A_\sigma$ is a stable matrix.
\item[(ii)] (on-neuron conditions)  $A_\sigma^{-1} 1_\sigma > 0$ and $a_\sigma> 1.$ 
\item[(iii)]  (off-neuron conditions) 
$\displaystyle{\sum_{i,j \in \sigma} A_{ki}(A_\sigma^{-1})_{ij} < 1}$ for each $k \notin \sigma$.
\end{itemize}
\end{theorem}

\begin{proof}
If $W$ is inhibitory and $\theta \leq 0$, then the only possible fixed point of~\eqref{eq:network} is $x^*=0$,
because the fixed point equations are $x^* = [Wx^* + \theta]_+$, and $Wx^* + \theta \leq 0$ for any $x^* \geq 0$.
By Proposition~\ref{prop:singleton}, this fixed point is guaranteed to be stable for $\theta < 0$.  Since $W$ is inhibitory, however, $x^* = 0$ is also a stable fixed point for $\theta = 0$.

To show the remaining statements, we fix nonempty $\sigma \subseteq [n]$, where $A_\sigma$ is not fine-tuned, and apply Theorem~\ref{thm:fixed-pts}.
It follows from the arguments above that if a stable fixed point $x^*$ with support $\sigma$ exists, then $\theta > 0$.  Condition (i) follows directly from condition (i) of Theorem~\ref{thm:fixed-pts}.  Since $\theta > 0$, condition (ii,b) from Theorem~\ref{thm:fixed-pts} does not apply, so it suffices to consider condition (ii,a), which splits into $a_\sigma < 0$ and $a_\sigma>1$ cases.  The remainder of this proof consists of showing that the $a_\sigma < 0$ case does not apply, yielding the appropriate specializations of conditions (ii) and (iii) above.

 Suppose $a_\sigma < 0$.  Then conditions (ii) and (iii) of Theorem~\ref{thm:fixed-pts} cannot be simultaneously satisfied, since condition (ii) requires $A_\sigma^{-1}1_\sigma<0$ and
 condition (iii) simplifies to 
$$\sum_{i,j \in \sigma} A_{ki}(A_\sigma^{-1})_{ij} > 1, \quad \text{for each} \;\; k \notin\sigma.$$
To see the contradiction, observe that $(A_\sigma^{-1}1_\sigma)_i = \sum_{j \in \sigma} (A_\sigma^{-1})_{ij}<0$ by condition (ii),
and so
$$\sum_{i,j \in \sigma} A_{ki}(A_\sigma^{-1})_{ij} = \sum_{i\in \sigma} A_{ki} \sum_{j \in \sigma} (A_\sigma^{-1})_{ij} \leq 0,$$
since $A_{ki} = W_{ki} + 1 \geq 0$, by assumption.
We can thus eliminate the $a_\sigma<0$ case from condition (ii,a), as only the $a_\sigma > 1$ case can apply.  Finally, for $\theta>0$ and $a_\sigma>1$ we see that condition (iii) of Theorem~\ref{thm:fixed-pts} simplifies to $\sum_{i,j \in \sigma} A_{ki}(A_\sigma^{-1})_{ij} < 1.$
\end{proof}

\begin{remark}\normalfont
Like Theorem~\ref{thm:fixed-pts}, Theorem~\ref{Thm0} does not necessarily cover the case of fixed points supported on a single neuron, $\sigma = \{i\}$, when $A_{ii} = 0$.  This case is covered by Proposition~\ref{prop:singleton}. 
\end{remark}

\subsection{Symmetric fixed point conditions}\label{sec:symm-fixed-point}

As in the inhibitory case above, the fixed point conditions of Theorem~\ref{thm:fixed-pts} can also be stated in simpler terms when $W$ is {\it symmetric with zero diagonal}, so that $W_{ij} = W_{ji}$ and $W_{ii}=0$ for all $i,j \in [n]$.
Note that this implies the auxiliary matrix $A$, given by~\eqref{eq:A}, is also symmetric with zero diagonal.

\begin{theorem}\label{Thm1}
Consider the threshold-linear network~\eqref{eq:network}, for $W$ (and hence $A$) symmetric with zero diagonal.  If $\theta < 0$, then $x^*=0$ is the unique stable fixed point of~\eqref{eq:network}.  If $\sigma \subseteq [n]$ is nonempty and 
$A_\sigma$ is not fine-tuned, then there exists a stable fixed point with support $\sigma$ if and only if $\theta > 0$ and the following three conditions hold:
\begin{itemize}
\item[(i)]  (permitted set condition) $-11^T + A_\sigma$ is a stable matrix.
\item[(ii)] (on-neuron conditions)  $A_\sigma^{-1} 1_\sigma > 0$.
\item[(iii)]  (off-neuron conditions) 
$\displaystyle{\sum_{i,j \in \sigma} A_{ki}(A_\sigma^{-1})_{ij} < 1}$ for each $k \notin \sigma$.
\end{itemize}
Moreover, if a stable fixed point $x^*$ with $\supp(x^*)=\sigma$ exists, then $a_\sigma > 1$.
\end{theorem}

\begin{proof}
The theorem consists of three statements, which we prove in reverse order.
First, recall from Lemma~\ref{lemma:geometric-char} that if $-11^T + A_\sigma$ is a stable matrix, then $a_\sigma > 1$, yielding the final statement.  
We now turn to the second statement, and show that it is a special case of Theorem~\ref{thm:fixed-pts}.  Note that we can reduce to the $a_\sigma>1$ case in condition (ii) of Theorem~\ref{thm:fixed-pts}, which only applies if $\theta>0$.  Since we must have $a_\sigma>1$ and $\theta>0$ in order for the off-neuron conditions to be relevant, we see that condition (iii) above is the correct simplification of condition (iii) in Theorem~\ref{thm:fixed-pts}.  
Finally, note that for $\theta < 0$ we cannot have a stable fixed point supported on a nonempty $\sigma$, by the arguments above.  By Proposition~\ref{prop:singleton}, the fixed point $x^* = 0$ with empty support exists and is stable in this case, giving us the first statement.
\end{proof}

Theorem~\ref{Thm1} does not cover the singleton case, $\sigma = \{i\}$, since in this case the matrix $A_\sigma$ is not invertible (because $A_{ii} = 0$) and is thus fine-tuned.  The following variant of Theorem~\ref{Thm1}, which we will use in our proof of Theorem~\ref{Thm2}, does include the singleton case.  It is a consequence of Theorem~\ref{Thm1}, Lemma~\ref{lemma:geometric-char}, Proposition~\ref{prop:singleton} and the proof of Theorem~\ref{thm:fixed-pts} (in order to drop the fine-tuned hypothesis).  Note that Proposition~\ref{prop:symm-fixed-points} is not an ``if and only if'' statement; the conditions here are only the necessary consequences of the existence of a stable fixed point.

\begin{proposition}\label{prop:symm-fixed-points}
Consider the threshold-linear network~\eqref{eq:network}, for $W$ (and hence $A$) symmetric with zero diagonal.  Let $\sigma \subseteq [n]$ be nonempty.  If there exists a stable fixed point with support $\sigma$, then $\theta > 0$ and the following three conditions hold:
\begin{itemize}
\item[(i)] $A_\sigma$ is a nondegenerate square distance matrix and $a_\sigma >1$.
\item[(ii)] $A_\sigma^{-1} 1_\sigma > 0$ if $|\sigma|>1$.
\item[(iii)]  If $|\sigma| > 1$, then
$\displaystyle{\sum_{i,j \in \sigma} A_{ki}(A_\sigma^{-1})_{ij} \leq 1}$ for each $k \notin \sigma$.\\
If $\sigma=\{i\}$, then $A_{ki} \leq 0$ for all $k \neq i$.
\end{itemize}
\end{proposition}
\begin{proof}
Suppose $\sigma$ supports a stable fixed point. Then $\sigma$ must be a permitted set, and so condition (i) follows directly from Lemma~\ref{lemma:geometric-char}.  (Recall our convention from Section~\ref{sec:geometry-background} that $a_\sigma = \infty$ if $\sigma = \{i\}$ and $A_{ii} = 0$.)
Next, observe that the $\theta>0$ requirement in Theorem~\ref{Thm1} holds even if $A_\sigma$ is fine-tuned, because it only required invertibility of $A_\sigma$ in the proof of Theorem~\ref{thm:fixed-pts}.   Since, by condition (i), $A_\sigma$ is a nondegenerate square distance matrix, the only instance when it is not invertible is when $|\sigma|=1$.  In this case, however, 
Proposition~\ref{prop:singleton} implies that $\theta>0$.

For the remaining conditions, we split into two cases: $|\sigma|>1$ and $|\sigma|=1$.
If $|\sigma|>1$, then $A_\sigma$ is invertible, and so $A_\sigma^{-1}1_\sigma>0$ (see the proof of Theorem~\ref{thm:fixed-pts}), yielding condition (ii).  In the proof of the forwards direction of Theorem~\ref{thm:fixed-pts}, we see that condition (iii) holds without the strict inequality even when $A_\sigma$ is fine-tuned, and so the simplified version of this condition in Theorem~\ref{Thm1} also holds, without the strict inequality, for all $A_\sigma$.  This gives us the first part of condition (iii).

If $\sigma=\{i\}$, Proposition~\ref{prop:singleton} provides the on-neuron condition $\theta>0$, which is automatically satisfied, so there is no further addition to condition (ii).  It also gives us the off-neuron condition $A_{ki} \leq A_{ii}$ for all $k \neq i$, which is the rest of condition (iii), since $A_{ii} = 0$.
\end{proof}

\section{Some geometric lemmas and Proposition~\ref{prop:antichain}}\label{sec:geometry}

In addition to Proposition~\ref{prop:symm-fixed-points}, the other main ingredient we will need to prove Theorem~\ref{Thm2} is the following technical result about nondegenerate square distance matrices:

\begin{proposition}\label{prop:antichain}
Let $A$ be an $n \times n$ matrix, and let $\tau \subseteq [n]$ with $|\tau| > 1$.  If $A_\tau$ is a nondegenerate square distance matrix  satisfying $A_\tau^{-1} 1_\tau > 0$, then for any $\sigma \subsetneq \tau$ with $|\sigma|>1$ there exists $k \in \tau\setminus \sigma$ such that $\sum_{i,j \in \sigma} A_{ki} (A_\sigma^{-1})_{ij} > 1.$
\end{proposition}

\noindent Proposition~\ref{prop:antichain} is key to the proof of Theorem~\ref{Thm2}  because supports of stable fixed points in the symmetric case correspond to nondegenerate square distance matrices $A_\sigma$ (see Lemma~\ref{lemma:geometric-char}).  Recalling Proposition~\ref{prop:symm-fixed-points}, we see that Proposition~\ref{prop:antichain} implies an incompatibility between the fixed point conditions for nested pairs of supports $\sigma,\tau \subseteq [n]$, with $\sigma \subsetneq \tau$.  Specifically, if $\tau$ satisfies both the permitted set condition (i) and the on-neuron conditions (ii) in Proposition~\ref{prop:symm-fixed-points}, then Proposition~\ref{prop:antichain} implies that the off-neuron conditions (iii) cannot hold for a proper subset $\sigma$.

The remainder of this section is devoted to the proof of Proposition~\ref{prop:antichain}.  Note that the proposition assumes $|\sigma|,|\tau|>1$, so that the symmetric matrices $A_\sigma$ and $A_\tau$ are invertible.  We will maintain this assumption throughout this section.

First, we need four geometric lemmas about nondegenerate square distance matrices.  
These lemmas rely on observations involving moments of inertia and center of mass for a particular mass configuration associated to the square distance matrix.  
Let $q_1, \ldots, q_n \in \RR^d$ be a configuration of $n$ points, with masses $y_1, \ldots, y_n$ assigned to each point, respectively.  Recall that the center of mass of this configuration is given by
$$q_{\cm} = \dfrac{1}{m(y)} \sum_{i = 1}^n y_i q_i,$$
where $m(y) = \sum_{i=1}^n y_i$ is the total mass.  The moment of inertia about any point $q \in \RR^d$ is 
$$I_q(y) = \sum_{i=1}^n \Vert q - q_i \Vert^2 y_i.$$

The classical parallel axis theorem states that the moment of inertia depends on the distance between the chosen point $q$ and the center of mass $q_{\cm}$.

\begin{proposition}[Parallel axis theorem]  \label{prop:parallel-axis}
Let $y_1,\ldots,y_n \in \RR$ denote the masses assigned to the points $q_1,\ldots, q_n \in \RR^d$. 
For any $q \in \RR^d$, 
$$I_{q}(y) = m(y) \Vert q - q_{\cm} \Vert^2 + I_{q_{\cm}}(y).$$
\end{proposition}

\noindent This result is well known, and is a staple of undergraduate physics.  It has also been called Appolonius' formula in Euclidean geometry \cite[Section 9.7.6]{Berger}.
For completeness, we provide a proof in Appendix~\ref{sec:gen-parallel-axis}, as a corollary of a novel and more general result.  From our proof, it is clear that Proposition~\ref{prop:parallel-axis} is valid for general dimension $d$ and general mass assignments, including negative masses.

We will exploit Proposition~\ref{prop:parallel-axis} by applying it to point configurations corresponding to nondegenerate square distance matrices, $A_\sigma$.
Let $\{q_i\}_{i \in \sigma}$  be a representing point configuration of $A_\sigma$ in $\RR^{|\sigma|-1}$, so that 
$$A_{ij} = \Vert q_i - q_j \Vert^2 \quad \text{for all } i,j \in \sigma.$$
For this point configuration, let $\{y_i\}_{i\in\sigma}$ be an assignment of masses to each point (not necessarily nonnegative), given by
$$y_i = (A_\sigma^{-1} 1_\sigma)_i \quad \text{for each } i \in \sigma.$$
Recall that because $A_\sigma$ is a nondegenerate square distance matrix and $|\sigma|>1$, $A_\sigma$ is invertible and hence these masses are always well defined.  The total mass is 
$$m(y) = \sum_{i\in\sigma} y_i = \sum_{i,j \in \sigma} (A_\sigma^{-1})_{ij} = a_\sigma.$$
Finally, note that because $\{q_i\}_{i \in \sigma}$ is a nondegenerate point configuration, there exists a unique equidistant point, $p_\sigma \in \RR^{|\sigma|-1},$ so that the distances $\Vert q_i - p_\sigma \Vert$ are the same for all $i \in \sigma$ \cite[Section 9.7.5]{Berger}. The following lemma shows us that the equidistant point coincides with the center of mass for this mass configuration.

\begin{lemma}  \label{lemma:qcm}
Let $A_\sigma$ be a nondegenerate square distance matrix, and assign masses $y_\sigma = A_\sigma^{-1}1_\sigma$ to a representing point configuration $\{q_i\}_{i\in\sigma}$, as described above.  Then
$q_{\cm} = p_\sigma.$
\end{lemma}

\begin{proof}
First, observe that the moment of inertia about any of the points $q_i$ is given by
$$I_{q_i}(y) = \sum_{j \in \sigma} \Vert q_i - q_j \Vert^2 y_j = (A_\sigma y_\sigma)_i = 
(A_\sigma A_\sigma^{-1} 1_\sigma)_i = 1, \quad \text{for each } i \in \sigma.$$
On the other hand, the parallel axis theorem tells us that
$I_{q_i}(y) = m(y)\Vert q_i - q_{\cm} \Vert^2 + I_{q_{\cm}}(y),$
so we must have that $\Vert q_i - q_{\cm} \Vert^2 = \Vert q_j - q_{\cm} \Vert^2$ 
for all $i,j \in \sigma.$  Clearly, these equations imply that $q_{\cm} = p_\sigma.$
\end{proof}

For a given point configuration $\{q_i\}_{i\in\sigma},$ we denote the interior of the convex hull by
$$(\conv{q_i}_{i\in\sigma})^o.$$  The next lemma states that the equidistant point $p_\sigma$ is inside the convex hull of its corresponding point configuration if and only all entries of $A_\sigma^{-1}1_\sigma$ are strictly positive.

\begin{lemma}\label{lemma:convex-hull}
Let $A_\sigma$ be a nondegenerate square distance matrix. Then $A_\sigma^{-1}1_\sigma>0$
if and only if $p_\sigma \in (\conv{q_i}_{i\in\sigma})^o$ for any representing point configuration 
$\{q_i\}_{i\in\sigma}.$
\end{lemma}

\begin{proof}
Let $\{q_i\}_{i\in\sigma}$ be a representing point configuration for $A_\sigma$, and assign masses $y_\sigma = A_\sigma^{-1}1_\sigma$, as in Lemma~\ref{lemma:qcm}, so that $q_\cm = p_\sigma.$  Next, observe that
$q_\cm \in \conv{q_i}_{i\in\sigma}$ if and only if the masses are all nonnegative.  Similarly, $q_\cm \in (\conv{q_i}_{i\in\sigma})^o$ if and only if the masses are all strictly positive.  It thus follows that
$p_\sigma \in (\conv{q_i}_{i\in\sigma})^o$ if and only if $A_\sigma^{-1}1_\sigma>0$.
\end{proof}

\begin{lemma}\label{lemma:qk}
Let $A_\tau$ be a nondegenerate square distance matrix with representing point configuration $\{q_i\}_{i\in\tau}$.  Suppose $\sigma \subsetneq \tau$, with $|\sigma|>1$, and let $p_\sigma$ denote the unique equidistant point in the affine subspace defined by the points $\{q_i\}_{i\in\sigma}$.  Then for any $k \in \tau \setminus \sigma$, we have $\sum_{i,j\in\sigma} A_{ki}(A_\sigma^{-1})_{ij} \leq 1$ if and only
if $\Vert q_k - p_\sigma \Vert \leq \Vert q_i - p_\sigma \Vert,$ for $i \in \sigma$.
\end{lemma}

\begin{proof}
First, assign masses $y_\sigma = A_\sigma^{-1}1_\sigma$ to the subset of points $\{q_i\}_{i\in\sigma}$, and observe that
$$\sum_{i,j\in\sigma} A_{ki}(A_\sigma^{-1})_{ij} = \sum_{i\in\sigma} A_{ki}(A_\sigma^{-1}1_\sigma)_i
= \sum_{i\in\sigma} \Vert q_k-q_i\Vert^2 y_i = I_{q_k}(y).$$
Next, recall from the proof of Lemma~\ref{lemma:qcm} that for each $i \in \sigma$, $I_{q_i}(y) = 1.$ 
Using Proposition~\ref{prop:parallel-axis}, we see that $I_{q_k}(y) \leq 1$ if and only if
$\Vert q_k - q_{\cm} \Vert \leq \Vert q_i - q_{\cm} \Vert$.  On the other hand, Lemma~\ref{lemma:qcm} tells us that $q_{\cm} = p_\sigma$, yielding the desired result.
\end{proof}

By definition of $p_\sigma$, the distances $\Vert q_i - p_\sigma \Vert$ for $i \in \sigma$ in the statement
of Lemma~\ref{lemma:qk} are all equal.  This distance to the equidistant point, $\rho_\sigma = \Vert q_i - p_\sigma \Vert$,
has in fact a well-known formula in terms of the matrix $A_\sigma$ (see \cite[Appendix C]{net-encoding}).
\begin{equation}\label{eqn:rho_sigma}
\rho_\sigma^2 = -\dfrac{1}{2}\dfrac{\det(A_\sigma)}{\cm(A_\sigma)} = \dfrac{1}{2a_\sigma}.
\end{equation}
To state the next lemma, we define
$$B_{\rho_\sigma}(p_\sigma) \od \{q \mid \Vert q - p_\sigma \Vert \leq \rho_\sigma\},$$
which is the closed ball of radius $\rho_\sigma$ centered at $p_\sigma$.  

\begin{figure}[h!]  
\begin{center}
\hspace{-.2in}
\includegraphics[width=5.25in]{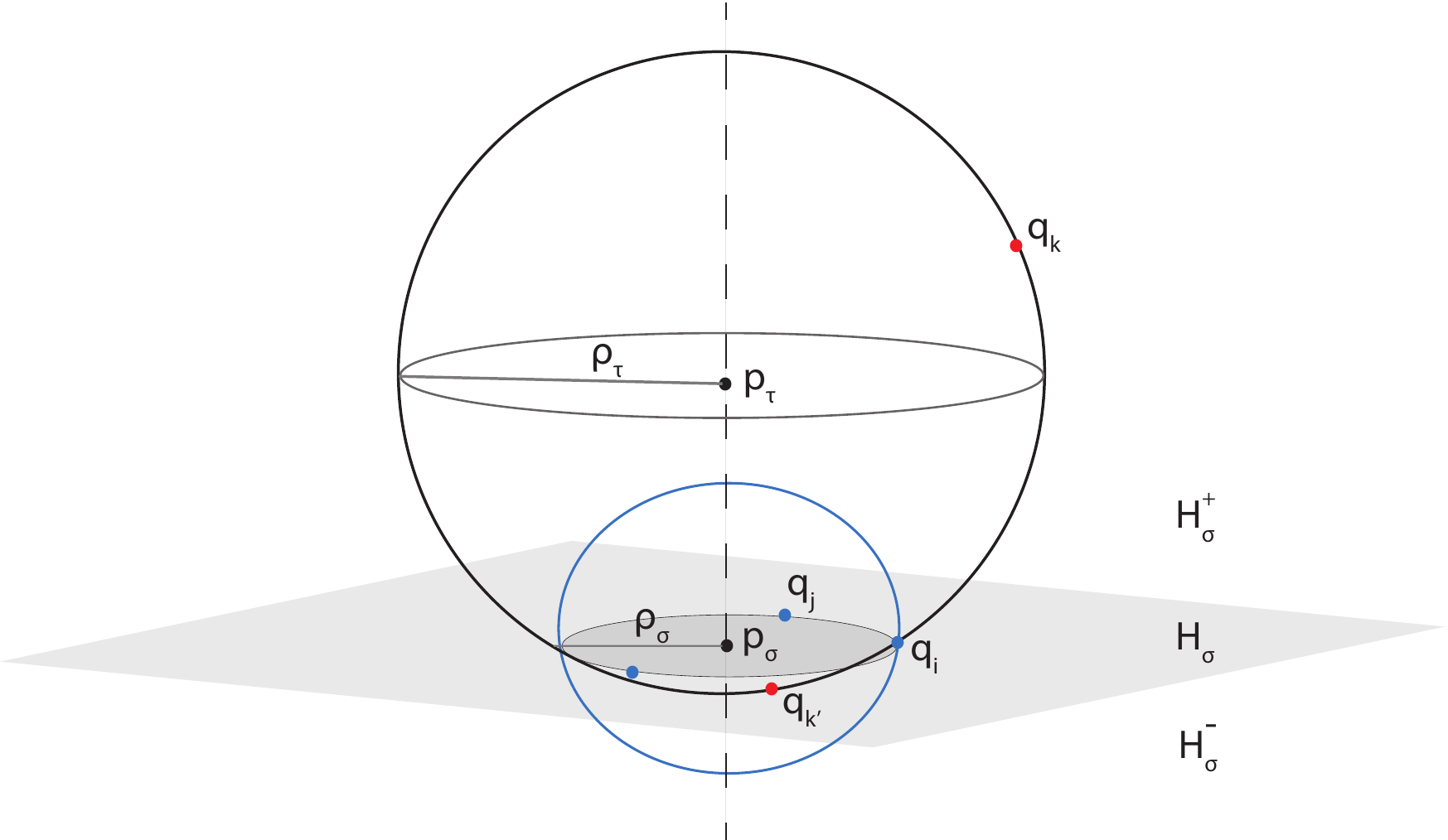}
\end{center}
\caption{Picture for the proof of Lemma~\ref{lemma:hierarchy}.}
\label{fig:spheres}
\end{figure}

\begin{lemma}\label{lemma:hierarchy}
Let $A_\tau$ be a nondegenerate square distance matrix with representing point configuration $\{q_i\}_{i\in\tau}$.  If $p_\tau \in (\conv{q_i}_{i\in\tau})^o$, and $\sigma \subsetneq \tau$ with $|\sigma|>1$, 
then there exists $k \in \tau \setminus \sigma$ such that $q_k \notin B_{\rho_\sigma}(p_\sigma).$
\end{lemma} 

\begin{proof}
Note that since $A_\tau$ is a nondegenerate square distance matrix, so is $A_\sigma$.  Let $p_\tau$ and $p_\sigma$ be the corresponding unique equidistant points, both embedded in $\RR^{|\tau|-1}$ (so $p_\sigma$ lies in the affine subspace spanned by $\{q_i\}_{i\in\sigma}$).  Let $H_\sigma$ denote the hyperplane in $\RR^{|\tau|-1}$ that contains
$p_\sigma$ and is perpendicular to the line $\overline{p_\sigma p_\tau}$ (see Figure~\ref{fig:spheres}).  Note that because $p_\tau$ is also equidistant to all $\{q_i\}_{i \in \sigma}$, the line
$\overline{p_\sigma p_\tau}$ is perpendicular to the affine subspace defined by  $\{q_i\}_{i \in \sigma}$; hence $H_\sigma$ contains all points $\{q_i\}_{i \in \sigma}$.
Denote by $H_\sigma^+$ the closed halfspace that contains 
$p_\tau$, and let $H_\sigma^-$ denote the opposite halfspace, so that $H_\sigma^+ \cap H_\sigma^- = H_\sigma$.   Because $p_\tau \in H_\sigma^+$ and $p_\tau \in (\conv{q_i}_{i \in \tau})^o$, we must have at least one $k \in \tau\setminus\sigma$ with $q_k \in H_\sigma^+\setminus H_\sigma$.  This implies $q_k \notin B_{\rho_\sigma}(p_\sigma)$, since the only way a point can lie both on the large sphere of radius $\rho_\tau$ and in the ball $B_{\rho_\sigma}(p_\sigma)$ is if it lies in $H_\sigma^-$ (see $q_{k'}$ in Figure~\ref{fig:spheres}).
\end{proof}

Finally, we are ready to prove Proposition~\ref{prop:antichain}.

\begin{proof}[Proof of  Proposition~\ref{prop:antichain}]
Suppose $A_\tau$ is a nondegenerate square distance matrix, satisfying $A_\tau^{-1} 1_\tau>0$. Let 
$\{q_i\}_{i \in \tau}$ be a representing point configuration.  
By Lemma~\ref{lemma:convex-hull},  we know that $p_\tau \in (\conv{q_i}_{i \in \tau})^o$.  Now applying
Lemma~\ref{lemma:hierarchy}, we find there exists $k \in \tau \setminus \sigma$ such that 
$q_k \notin B_{\rho_\sigma}(p_\sigma).$  This means $\Vert q_k - p_\sigma \Vert > \Vert q_i - p_\sigma \Vert$ for $i \in \sigma$, which in turn implies that $\sum_{i,j\in\sigma} A_{ki}(A_\sigma^{-1})_{ij} > 1$, by Lemma~\ref{lemma:qk}.
\end{proof}

\section{Proof of Theorems~\ref{Thm2} and~\ref{thm:symmetric-CTLN}}\label{sec:main-proofs}

In this section we prove Theorems~\ref{Thm2} and~\ref{thm:symmetric-CTLN}.  Following these proofs, we state and prove a related result, Theorem~\ref{Thm3}.

The proof of Theorem~\ref{Thm2} follows from Propositions~\ref{prop:symm-fixed-points} and~\ref{prop:antichain}, together with a simple lemma (below).
Recall from Section~\ref{sec:prelim-permitted} that for a given connectivity matrix $W$, $\P(W)$ denotes the set of all permitted sets of the corresponding threshold-linear network.  We will use the notation $\P_{\max}(W)$ to denote the {\it maximal} permitted sets, with respect to inclusion.  In other words, if $\sigma \in \P_{\max}(W)$, then $\tau \notin \P(W)$ for any $\tau \supsetneq \sigma.$

\begin{lemma}\label{lemma:sym-singleton}
Consider the threshold-linear network~\eqref{eq:network}, for $W$ symmetric with zero diagonal.  If there exists a stable fixed point $x^*$ with support $\sigma = \{i\}$, then $\sigma \in \P_{\max}(W)$.
\end{lemma}

\begin{proof}
Suppose $\sigma = \{i\}$ is not a maximal permitted set, so that $\sigma \subsetneq \tau$ for some $\tau \in \P(W)$.  Since $A_\tau$ must be a nondegenerate square distance matrix, and $|\tau| \geq 2$, it follows that $A_{ki}>0$ for all $k \in \tau \setminus \{i\}$.  But this contradicts condition (iii) of Proposition~\ref{prop:symm-fixed-points}, so we can conclude that $\sigma \in \P_{\max}(W).$
\end{proof}

Recall that Proposition~\ref{prop:antichain} implies that if $\tau$ satisfies conditions (i) and (ii) of Proposition~\ref{prop:symm-fixed-points}, then no $\sigma \subsetneq \tau$ with $|\sigma|>1$ can possibly satisfy condition (iii).  
Lemma~\ref{lemma:sym-singleton} allows us to extend this observation to the case of $|\sigma|=1$, as singletons can only support stable fixed points if they are maximal permitted sets.  Proposition~\ref{prop:symm-fixed-points} also tells us that the existence of a stable fixed point with nonempty support $\tau$ implies $\theta>0$, ruling out the possibility of the stable fixed point with empty support, $x^*=0$.  Thus, the above consequence of Proposition~\ref{prop:antichain} can be extended to all $\sigma \subsetneq \tau$.  We are now ready to prove Theorem~\ref{Thm2}

\begin{proof}[Proof of Theorem~\ref{Thm2}]
Suppose there exists a stable fixed point with support $\tau$.  We will first show
that for $\sigma \subsetneq \tau$, $\sigma$ cannot support a stable fixed point.  We may assume $\tau$ is nonempty (otherwise it has no proper subsets).  It follows from Proposition~\ref{prop:symm-fixed-points} that $\theta>0$.
We consider three cases: $\sigma = \emptyset$, $|\sigma| = 1$, and $|\sigma| \geq 2$.

Suppose $\sigma = \emptyset$.  Since $\theta>0$, we know that $\sigma$ cannot support a stable fixed point, by Proposition~\ref{prop:singleton}.  Next, suppose $|\sigma|=1$.  Here we can also conclude that 
$\sigma$ cannot support a stable fixed point, by Lemma~\ref{lemma:sym-singleton}.  (Note that $\sigma\notin \P_{\max}(W)$ because $\sigma \subsetneq \tau$, and we must have $\tau \in \P(W)$.)  Finally, suppose $|\sigma| \geq 2$.  Since both $|\sigma|,|\tau|>1$, we can apply Proposition~\ref{prop:antichain}.  All hypotheses are satisfied because $\tau$ is the nonempty support of a stable fixed point, and so by Proposition~\ref{prop:symm-fixed-points} we know that $A_\tau$ is a nondegenerate square distance matrix with $A_\tau^{-1} 1_\tau > 0$.  It thus follows from Proposition~\ref{prop:antichain}
that $\sigma$ cannot satisfy condition (iii) of Proposition~\ref{prop:symm-fixed-points}, so $\sigma$ cannot be a stable fixed point support.

Finally, observe that for any proper superset $\sigma \supsetneq \tau$, if $\sigma$ is the support of a stable fixed point, then the above logic implies that $\tau$ cannot support a stable fixed point, contradicting the hypothesis.  Thus, there is no stable fixed point with support $\sigma$ for any $\sigma \subsetneq \tau$ or $\sigma \supsetneq \tau$.  
\end{proof}

We now turn to the proof of Theorem~\ref{thm:symmetric-CTLN}.  We will use the following simple lemma.

\begin{lemma}\label{lemma:A_sigma}
Let $G$ be a simple graph, and consider the binary symmetric network $W=W(G, \varepsilon, \delta)$, for some $0 < \varepsilon<1$ and $\delta >0$, as in~\eqref{eq:binary-synapse}.  Let $A=11^T-I+W$, as in~\eqref{eq:A}.  If $\sigma$ is a clique of $G$ with $|\sigma|>1$, then 
\begin{enumerate}
\item[(i)] $-11^T+A_\sigma$ is stable, and 
\item[(ii)] $A_{\sigma}^{-1} 1_\sigma = \dfrac{1}{\varepsilon(|\sigma| -1)}1_\sigma$ (and thus $a_\sigma =  \dfrac{|\sigma|}{\varepsilon(|\sigma| -1)}$).
\end{enumerate}
\end{lemma}
\begin{proof}  We first prove (i), and then (ii).

(i) Observe that if $\sigma$ is a clique and $|\sigma|>1$, then $-11^T+A_{\sigma}$ is a matrix with all diagonal entries equal to $-1$ and all off-diagonal entries equal to $-1+\varepsilon$.  Thus, $-11^T+A_{\sigma} = (-1+\varepsilon)11^T - \varepsilon I_\sigma$.  The term $(-1+\varepsilon)11^T$ is a rank one matrix, whose only non-zero eigenvalue is $|\sigma| (-1+\varepsilon)$, corresponding to the eigenvector $1_\sigma$.  Thus, the eigenvalues of $-11^T+A_{\sigma}$ are $|\sigma|(-1+\varepsilon) -\varepsilon$ and $-\varepsilon$.  The eigenvalue $|\sigma|(-1+\varepsilon) -\varepsilon$ is negative precisely when $\varepsilon < |\sigma|/(|\sigma|-1)$.  This always holds since $|\sigma|>1$ and $\varepsilon<1$.   We conclude that all eigenvalues of $-11^T+A_{\sigma}$ are negative, and hence $-11^T+A_{\sigma}$ is stable.  

(ii) Notice that since $\sigma$ is a clique, $A_\sigma = \varepsilon(11^T-I)_\sigma$.  It is easy to check that $A_\sigma^{-1}$ satisfies,
\[
A_\sigma^{-1} = \dfrac{1}{\varepsilon(|\sigma|-1)} {\scriptsize \begin{bmatrix} -|\sigma|+2 & 1& 1 & \ldots & 1\\ 1 & -|\sigma|+2 & 1 & \ldots & 1\\ 1 & 1& -|\sigma|+2& \ldots & 1 \\ \vdots & \vdots & &\ddots  &\vdots \\ 1&1&1&\ldots &-|\sigma|+2 \end{bmatrix} },
\]
from which it immediately follows that $A_{\sigma}^{-1} 1_\sigma = \dfrac{1}{\varepsilon(|\sigma| -1)}1_\sigma$.
\end{proof}

\begin{proof}[Proof of Theorem~\ref{thm:symmetric-CTLN}]
First we show that non-cliques cannot support stable fixed points, because they are not permitted sets.  Then we show that all maximal cliques do support stable fixed points, and derive the corresponding expression for $x_\sigma^*$.  Finally, we invoke Theorem~\ref{Thm2} to conclude that there can be no other stable fixed points.

Let $\sigma \subseteq [n]$ be a subset that is not a clique in $G$. 
In the case $\sigma = \emptyset$, we know from Proposition~\ref{prop:singleton} that there is no $x^*=0$ fixed point because $\theta>0.$   Suppose, then, that $\sigma$ is nonempty.  Since single vertices are always cliques, we know that $|\sigma| \geq 2$ and there exists at least one pair $i,j \in \sigma$ such that $(ij) \notin G$.  The corresponding principal submatrix $(-11^T+A_{\sigma})_{\{ij\}}$ is given by $\scriptsize {\begin{bmatrix} -1 & -1-\delta \\ -1-\delta & -1 \end{bmatrix} } $.   Since both the determinant and trace of this submatrix are negative, it must have a positive eigenvalue and is thus unstable.  By the Cauchy Interlacing Theorem (see \cite[Appendix A]{net-encoding} or \cite{cauchy}), $-11^T+A_{\sigma}$ must also have a positive eigenvalue, rendering it unstable.  It follows from~\eqref{eq:P(W)} that $\sigma$ is not a permitted set, so there is no stable fixed point with support $\sigma$.  

Next, we show that all maximal cliques support stable fixed points.  We split into two cases: $|\sigma|=1$ and $|\sigma|>1$.
Let  $\sigma=\{i\}$ be a maximal clique consisting of a single vertex, $i$.  Since $A_{ii}=0 <1$ and $\theta>0$, conditions (i) and (ii) of Proposition~\ref{prop:singleton} are always satisfied.  Since $\{i\}$ is a maximal clique, $(ik) \notin G$ for all $k\neq i$, and so $A_{ki}=-\delta < 0 = A_{ii}$ for all $k\neq i$.  Thus, condition (iii) of Proposition~\ref{prop:singleton} is also satisfied.  It follows that $\sigma = \{i\}$ supports a stable fixed point, given by $x_i^* = \theta$ and $x_k^* = 0$ for all $k \neq i$, in agreement with the desired formula for $x_\sigma^*$ with $|\sigma|=1$.

Now suppose $\sigma$ is a maximal clique with $|\sigma|>1$.  Note that $A_\sigma$ is not fine-tuned, so Theorem~\ref{Thm1} applies.  By  Lemma~\ref{lemma:A_sigma}, $-11^T+A_\sigma$ is stable and $A_{\sigma}^{-1} 1_\sigma = \dfrac{1}{\varepsilon(|\sigma| -1)}1_\sigma >0$.  Thus, conditions (i) and (ii) of Theorem~\ref{Thm1} are satisfied, and $\theta>0$ by assumption.  It remains only to show that condition (iii) holds.
Since $\sigma$ is a maximal clique, for all $k \notin \sigma$ there exists some $i_k \in \sigma$ such that $(i_kk) \notin G$, and thus $A_{k i_k}=-\delta$.  We obtain:
\[
\sum_{i,j \in \sigma} A_{ki} (A_{\sigma}^{-1})_{ij} = \sum_{i \in \sigma} A_{ki}\sum_{j \in \sigma}(A_{\sigma}^{-1})_{ij} =
\dfrac{1}{\varepsilon(|\sigma|-1)}\sum_{i \in \sigma} A_{ki} 
= \dfrac{-\delta + \sum_{i \in \sigma \setminus \{i_k\}} A_{ki}}{\varepsilon(|\sigma|-1)},
\]
where we have used part (ii) of Lemma~\ref{lemma:A_sigma} to evaluate the row sums of $A_{\sigma}^{-1}$.  Since $A_{ki} \leq \varepsilon$ for each $i \in \sigma$, it follows that $\sum_{i \in \sigma \setminus \{i_k\}} A_{ki} \leq \varepsilon(|\sigma|-1)$, and so condition (iii) is satisfied:
$$\sum_{i,j \in \sigma} A_{ki} (A_{\sigma}^{-1})_{ij} \leq  \dfrac{-\delta}{\varepsilon(|\sigma|-1)} + 1 < 1.$$
We conclude that there exists a stable fixed point for each maximal clique $\sigma$.  Using the formula for $x_\sigma^*$ from Theorem~\ref{thm:fixed-pts}, together with the expressions for $A_\sigma^{-1} 1_\sigma$ and $a_\sigma$ from Lemma~\ref{lemma:A_sigma}, we obtain the desired equation for the fixed point:
$$x_\sigma^* = \dfrac{\theta}{a_\sigma-1} A_\sigma^{-1} 1_\sigma = \dfrac{\theta}{(1-\varepsilon)|\sigma|+\varepsilon}1_\sigma.$$

Finally, observe that since any non-maximal clique is necessarily a proper subset of a maximal clique, Theorem~\ref{Thm2} guarantees that non-maximal cliques can not support stable fixed points.  Thus, the supports of stable fixed points correspond precisely to maximal cliques in $G$.
\end{proof}

In Theorem~\ref{thm:symmetric-CTLN}, we saw that all fixed point supports corresponded to maximal permitted sets, as these were the maximal cliques of the underlying graph $G$.  This situation is ideal for pattern completion, as it guarantees that only maximal patterns can be returned as outputs of the network.
Our final result shows that this phenomenon generalizes to a broader class of symmetric networks, provided all maximal permitted sets $\tau \in \P_{\max}(W)$ satisfy the on-neuron conditions, $A_\tau^{-1} 1_\tau>0,$ from Theorem~\ref{Thm1}.

\begin{theorem}\label{Thm3}
Let $W$ be symmetric with zero diagonal, and suppose that $\theta>0$ and $A_\tau^{-1} 1_\tau>0$ for each $\tau \in \P_{\max}(W)$ with $|\tau|>1$.  If $x^*$ is a stable fixed point of~\eqref{eq:network}, then $\supp(x^*)\in \P_{\max}(W)$.
\end{theorem}

\begin{proof}
Let $\sigma = \supp(x^*)$ be the support of a stable fixed point $x^*$.  Observe that $\sigma \neq \emptyset$, since Proposition~\ref{prop:singleton} states that $x^*=0$ can only be a fixed point when $\theta \leq 0$.  If $\sigma=\{i\}$ for some $i \in [n]$, then by Lemma~\ref{lemma:sym-singleton} we have that $\{i\} \in \P_{\max}(W)$.  

Next, consider $|\sigma|>1$.  Since $\sigma \in \P(W)$, there exists some $\tau \in \P_{\max}(W)$ such that $\sigma \subseteq \tau$.  Since $\tau$ is a permitted set, $A_\tau$ is a nondegenerate square distance matrix (Lemma~\ref{lemma:geometric-char}).  By hypothesis, $A_\tau^{-1} 1_\tau > 0$, so Proposition~\ref{prop:antichain} applies.  If $\sigma\subsetneq \tau$ is a proper subset, then there exists $k \in \tau\setminus \sigma$ that violates condition (iii) of Theorem~\ref{Thm1}, contradicting the assumption that $\sigma$ is a stable fixed point support.  It follows that $\sigma =\tau,$ and thus $\sigma \in \P_{\max}(W)$.  
\end{proof}

\section{Appendix}

\subsection{Additional details for the simulations in Section~\ref{sec:application}}\label{sec:gen-PF-codes}

\paragraph{Network implementation.}
We solved the system of differential equations~\eqref{eq:network}, with $\theta = 1$, using a standard Matlab ode solver for a length of time corresponding to $50\tau_{_\mathrm{ L}}$, where $\tau_{_\mathrm{L}}$ is the leak time constant associated to each neuron.  This length of time was sufficient for the network to numerically stabilize at a fixed point $x^*$.  Note that $\tau_{_\mathrm{L}}$ is omitted from our equations because we have set $\tau_{_\mathrm{L}} = 1,$ so that time is measured in units of $\tau_{_\mathrm{L}}$.

\paragraph{Generation of 2D place field codes.} 
Two-dimensional place field codes were generated following the same methods as in \cite{neuro-coding}.  200 place field centers were randomly chosen from a $1 \times 1$ square box environment, with each place field a disk of radius $0.15$. This produced place field codes with an average of 7\% of neurons firing per codeword.  As described in \cite{neuro-coding}, the place field centers were initially chosen randomly from 
uncovered regions of the stimulus space, until complete coverage was achieved.  The remaining place field centers were then chosen uniformly at random from the full space.  Here we introduced one modification to the procedure in \cite{neuro-coding}: our place field centers were generated 50 at a time, repeating the process from the beginning for the four sets of 50 neurons in order to guarantee that every point in the stimulus space was covered by a minimum of four place fields.  This ensured that all codewords had at least four 1s (out of 200 bits).

\subsection{A generalization of the parallel axis theorem}\label{sec:gen-parallel-axis}

Here we present a new, more general version of the parallel axis theorem (Proposition~\ref{prop:parallel-axis}), which was used in Section~\ref{sec:geometry}.

\begin{proposition}\label{prop:gen-parallel-axis}
Let $A$ be an $n \times n$ matrix, and $y \in \RR^n$ a vector such that $m = m(y)=\sum_{i=1}^n y_i \neq 0$.  Then there exists a unique $\lambda = (\lambda_1,\ldots,\lambda_n)$ such that
$$Ay = \Lambda y, \quad \text{ where } \quad \Lambda_{ij} = \lambda_i + \lambda_j.$$
Explicitly,
\begin{equation}\label{eq:lambda_i}
\lambda_i = \dfrac{(Ay)_i}{m} - \dfrac{ y \cdot (Ay)}{2m^2}.
\end{equation}
\end{proposition}

\begin{proof}
First, observe that $Ay = \Lambda y$ implies
$$(Ay)_i = (\Lambda y)_i = \sum_j (\lambda_i + \lambda_j) y_j = m \lambda_i + \lambda \cdot y.$$
Since $m \neq 0$, 
$\lambda_i = \dfrac{(Ay)_i - \lambda \cdot y}{m}.$
Taking the dot product with $y$ yields
$$\lambda \cdot y = \dfrac{1}{m}y \cdot (Ay) - \dfrac{1}{m}(\lambda \cdot y) \sum_i y_i 
= \dfrac{1}{m}y \cdot (Ay) - \lambda \cdot y,$$
allowing us to solve for
$\lambda \cdot y = \dfrac{1}{2m} y \cdot (Ay).$
Plugging this into the above expression for $\lambda_i$ yields the desired result.
\end{proof}

We can now obtain the classical parallel axis theorem as a special case of Proposition~\ref{prop:gen-parallel-axis}.
Consider the case where $A$ is the square distance matrix for a configuration of points $q_1,\ldots q_n \in \RR^{d}$, and $y = (y_1,\ldots,y_n)^T$ is a vector of masses, one for each point, whose sum $m = \sum_{i=1}^n y_i$ is nonzero.  In this case, $A_{ij} = \Vert q_i - q_j \Vert^2$, the center of mass is $q_{\cm} = \dfrac{1}{m}\sum_{i=1}^n q_iy_i$, and
it is not difficult to check that
$$\lambda_i  = \dfrac{(Ay)_i}{m} - \dfrac{ y \cdot (Ay)}{2m^2} = \Vert q_i - q_{\cm} \Vert^2.$$
(Without loss of generality, choose $q_{\cm} = 0$ and use this fact to cancel terms of the form $\sum_i q_i y_i$ that appear when you rewrite $A_{ij} = (q_i - q_j)\cdot(q_i - q_j)$ and expand.)  Recall that the moment of inertia of such a mass configuration about a point $q \in \RR^d$ is given by
$I_q(y) = \sum_{j=1}^n \Vert q - q_j \Vert^2 y_j.$
If $q = q_i$ for some $i \in [n]$, then we have
$$I_{q_i}(y) = \sum_{j=1}^n \Vert q_i - q_j \Vert^2 y_j =(Ay)_i.$$
We can now prove Proposition~\ref{prop:parallel-axis}, which states that:
$$I_{q}(y) = m \Vert q - q_{\cm} \Vert^2 + I_{q_{\cm}}(y).$$

\begin{proof}[Proof of Proposition~\ref{prop:parallel-axis}] 
Without loss of generality, we can assume that $q = q_i$ for some $i \in [n]$.  (If not, add the point $q$ to the collection $\{q_i\}$ and assign it a mass of $0$.)
As observed above, $I_q(y) = I_{q_i}(y) = (Ay)_i,$ and recall
from the proof of Proposition~\ref{prop:gen-parallel-axis} that $(Ay)_i = m\lambda_i + \lambda \cdot y$.   
Since $\lambda_i = \Vert q_i - q_{\cm} \Vert^2$, and thus $\lambda \cdot y = I_{q_{\cm}}(y),$ substituting these expressions into the equation for $(Ay)_i$
immediately yields the desired result.
\end{proof}

To see why Proposition~\ref{prop:gen-parallel-axis} may be useful more generally, consider the situation where $y\in \RR^n$ is proportional to the vector of firing rates at a fixed point of a threshold-linear network~\eqref{eq:network}.
Specifically, suppose $y$ has support $\sigma$, and $y_\sigma = A_\sigma^{-1}1_\sigma$.  In this case, 
$m = \sum_{i \in \sigma}  (A_\sigma^{-1}1_\sigma)_i = a_\sigma$, 
$(Ay)_i = 1$ for all $i \in \sigma$,
and $y \cdot (Ay) =  \sum_{i \in \sigma} y_i = m = a_\sigma$. 
It follows that
$$\lambda_i = \dfrac{1}{2a_\sigma} \;\; \text{ for all } \; \; i \in \sigma.$$
On the other hand, for $k \notin \sigma$ we have 
 $(Ay)_k = \sum_{i \in \sigma} A_{ki} y_i =\sum_{i,j \in \sigma} A_{ki} (A_\sigma^{-1})_{ij},$ so that
$$ \lambda_k = \dfrac{1}{a_\sigma}\left(\sum_{i,j\in\sigma}A_{ki}(A_\sigma^{-1})_{ij} - \dfrac{1}{2}\right) \;\; \text{ for all } \; \; k \notin \sigma.$$
Assuming $a_\sigma>0$, the off-neuron conditions of Theorem~\ref{Thm0} and Theorem~\ref{Thm1} are satisfied if and only if $\lambda_k < \dfrac{1}{2a_\sigma} = \lambda_i$.
Thus, $\lambda_i$ for $i \in \sigma$ generalizes the quantity
$$\Vert q_i - q_{\cm} \Vert^2 = \Vert q_i - p_{\sigma} \Vert^2 = \rho_\sigma^2 = \dfrac{1}{2a_\sigma},$$ 
which appeared in Section~\ref{sec:geometry} for $A$ a nondegenerate square distance matrix (see equation~\eqref{eqn:rho_sigma}).
Similarly, for $k \notin \sigma$, the condition $\lambda_k<\lambda_i$ generalizes the requirement 
$\Vert q_k - p_\sigma \Vert \leq \Vert q_i - p_\sigma \Vert$ from Lemma~\ref{lemma:qk}.  This suggests that it may be possible to generalize Proposition~\ref{prop:antichain}, 
and thus Theorem~\ref{Thm2}, beyond the symmetric case.

\section{Acknowledgments}
CC was supported by NSF DMS-1225666/DMS-1537228, NSF DMS-1516881, and an Alfred P. Sloan Research Fellowship.


\bibliographystyle{unsrt}
\bibliography{network-refs}

\end{document}